\documentclass{article}

\usepackage{enumerate}
\usepackage{hyperref}
\usepackage{color}

\usepackage{times}
\usepackage{helvet}
\usepackage{courier}
\usepackage{amsmath}
\usepackage{amssymb}
\usepackage{amsfonts}
\usepackage{amsthm}
\usepackage{mathrsfs}
\usepackage{xspace}
\usepackage{tikz}
\usetikzlibrary{arrows,automata,calc,matrix}
\usepackage[lined,plain,commentsnumbered]{algorithm2e}
\usepackage{verbatim}
\usepackage{enumitem}

\usepackage[lined,plain,commentsnumbered]{algorithm2e}

\newtheorem{definition}{Definition}
\newtheorem{proposition}{Proposition}
\newtheorem{theorem}{Theorem}
\newtheorem{thm}{Theorem}
\newtheorem{lemma}[thm]{Lemma}

\newtheorem{example}[thm]{Example}

\newcommand{\I}{\mathcal I}
\newcommand{\N}{\mathcal N}

\newcommand{\B}{\prof{B}}
\newcommand{\V}{P}
\newcommand{\vis}{V }

\newcommand{\Inf}{\mathit{Inf}}

\newcommand{\prof}[1]{\text{\boldmath $#1$}}

\newcommand{\hide}[1]{ \mathsf{hide} (#1) }
\newcommand{\reveal}[1]{ \mathsf{reveal}(#1) }

\newcommand{\nextop}[1]{ \mathsf{X} #1 }
\newcommand{\until}[2]{ #1 \mathsf{U} #2}
\newcommand{\know}[2]{ \mathsf{K}_{#1} #2}

\newcommand{\henceforth}[1]{ \mathsf{G} #1 }
\newcommand{\eventually}[1]{ \mathsf{F} #1 }

\newcommand{\bel}[2]{ \mathsf{op}( #1,#2 ) }
\newcommand{\visatom}[2]{ \mathsf{vis}( #1,#2 ) }
\newcommand{\pubatom}[2]{ \mathsf{pub}( #1,#2 ) }

\newcommand{\skipact}{ \mathsf{skip}}

\newcommand{\actions}{ \mathcal A}
\newcommand{\jactions}{ \mathcal{J}}

\newcommand{\lang}{ \mathcal{L}_{\logic}}
\newcommand{\logic}{ \mathsf{ELTL{\text{--}I}} }
\newcommand{\logicgoal}{ \mathsf{LTL{\text{--}I} }}
\newcommand{\langminus}{ \mathcal{L}_{\logicminus}}
\newcommand{\logicminus}{ \mathsf{LTL{\text{--}I}} }

\newcommand{\succfunct }{ \mathit{succ}}

\newcommand{\stateset}{ \mathcal{S}}

\newcommand{\histset}{ \mathcal{H}}

\newcommand{\bnf}{::=}
\newcommand{\infgame}{ \mathit{IG}}

\newcommand{\opset}{ \mathcal{B}}
\newcommand{\stratset}{ \mathcal{Q}}
\newcommand{\pubopset}{ \mathcal{P}}
\newcommand{\strat}{ \mathit{Q}}

\newcommand{\redu}{ \text{red}}

\begin{document}

% In the original styles from ACM, you would have needed to
% add meta-info here. This is not necessary for AAMAS 2015  as
% the complete copyright information is generated by the cls-files.

\title{Strategic disclosure of opinions\\ on a social network}

\author{Umberto Grandi, Emiliano Lorini and Laurent Perrussel\thanks{This work was partially supported by Labex CIMI (ANR-11-LABX-0040-CIMI), within the program ANR-11-IDEX-0002-02.}}

\date{IRIT and CNRS\\
University of Toulouse\\
France}

\maketitle

\begin{abstract}

\noindent
We study the strategic aspects of social influence in a society
of agents linked by a trust network, introducing a new class of games called games of influence. A game of influence is an infinite repeated game with incomplete information in which, at each stage of interaction, an agent can make
her opinions visible (public) or invisible (private) in order to influence other agents' opinions. The influence process is mediated by a trust network,
as we assume that the opinion of a given agent is only affected by the opinions of those agents that she considers trustworthy (i.e., the agents in the trust network that are directly linked to her). Each agent is endowed with a goal, expressed in a suitable temporal language inspired from linear temporal logic ($\mathsf{LTL}$). We show that games of influence provide a simple abstraction to
explore the effects of the trust network structure on the agents' behaviour, by considering solution concepts from game-theory such as Nash equilibrium, weak dominance and winning strategies.

\end{abstract}

%%%%%%%%%%%%%%%%%%%%%%%%%%
\section{Introduction}\label{sec:introduction}
%%%%%%%%%%%%%%%%%%%%%%%%%%

%At the micro-level, social influence can be conceived as the process by means of which
At the micro-level, social influence can be conceived as a process where an agent forms her opinion on the basis of the opinions expressed by other agents in the society. Social influence depends on trust since an agent can be influenced by another agent, so that her opinions are affected by the expressed opinions of the other, only if she trusts her.
At the macro-level, social influence is the basic mechanism driving the diffusion of opinions in human societies: certain agents in the society influence other agents in the society towards a given view, and these agents, in turn, influence other agents to acquire the same view, and so on. In other words, social influence can be seen as the driving force of opinion diffusion in human and human-like agent societies.
This view is resonant of existing studies in social sciences and social psychology which emphasize the role of interpersonal processes in how people construe and form their perceptions, judgments, and impressions (see, e.g., \cite{Asch,Festinger,Moscovici}).

Recent work in multi-agent systems \cite{SchwindEtAlAAAI2015, GrandiEtAlAAMAS2015} proposed a formal model of opinion diffusion that combined methods and techniques from social network analysis with methods and techniques from belief merging and judgment aggregation.
The two models % are aimed
aim at studying how opinions of agents on a given set of issues evolve over time due to the influence of other agents in the population.
The basic component of these models is the trust network, as it is assumed that the opinions of a certain agent are affected only by the opinions of the agents that she trusts (i.e., the agents in the trust network that are directly linked to her).
Specifically, the opinions of a certain agent at a given time are the result of aggregating the opinions of the trustworthy agents at the previous time.
%UG: I think we are being too specific here:
%Different aggregation criteria are considered such as unanimity and majority. Moreover, the conditions under which opinions of agents converge in the long-term, depending on the structure of the trust network, are studied (e.g.,  Do opinions converge in a complete graph? Do they converge in acyclic graphs?).

In this work we build on these models to look at social influence from a strategic perspective. We do so by introducing a new class of games, called games of influence.
%, that rely on the social influence mechanism studied in \cite{GrandiEtAlAAMAS2015}.
Specifically, a game of influence is an infinite repeated game with incomplete information in which, at each stage of interaction, an agent can make her opinions visible (public) or invisible (private) to the other agents. Incompleteness of information is determined by the fact that an agent has uncertainty about the private opinions of the other agents, as she cannot see them.
%Similarly to the model of social influence presented in \cite{GrandiEtAlAAMAS2015},
At each stage of the game, every agent is influenced by the \emph{public} opinions of the agents she trusts (i.e., her neighbors in the trust network) and changes her opinions on the basis of the aggregation criterion she uses.

Following the representation of agents' motivations given in \cite{GutierrezEtAlIC20015}, in a game of influence each agent is identified with the goal that she wants to achieve. This goal is represented by a formula of a variant of linear temporal logic ($\mathsf{LTL}$), in which we can express properties about agents' present and future opinions.
For example, an agent might have the achievement goal that at some point in the future there will be consensus about a certain proposition $p$ (i.e., either everybody has the opinion that $p$ is true or everybody has the opinion that $p$ is false), or the maintenance goal that two different agents will always have the same opinion about $p$.

Games of influence provide a simple abstraction to explore
the effects of the trust
network structure on the agents' behaviour.
We consider solution concepts
from game-theory such as Nash equilibrium,
weak dominance and winning strategies.
For instance,
in the context of
games of influence,
we can study
how
the relative position of an agent in the trust network determines her
influencing power,
that is,
her capacity to influence opinions
of other agents,
no matter what the others decide to do
(which corresponds to the concept
of uniform strategy).
Moreover, in games of influence one can study
how the structure of the trust network
determines existence of Nash equilibria,
depending on the form of the agents' goals.
For instance,
we will show
that
if the trust
network is fully connected
and every agent wants to reach a consensus about
a certain proposition $p$, then
there always exists a least one Nash equilibrium.

\subsection*{Related work and paper outline}

Apart from the above mentioned work on opinion diffusion via judgment aggregation \cite{GrandiEtAlAAMAS2015} and belief merging \cite{SchwindEtAlAAAI2015}, there is a vast interest in providing formal models of social influence.
The most relevant is probably the work of Gosh and Vel\'azquez-Quesada \cite{GoshQuesadaAAMAS2015}, which does not however consider strategic aspects in their preference update model.
The Facebook logic introduced by Seligman \emph{et al.} \cite{SeligmanEtAlTARK2013} is also relevant, and motivated our effort in Section~\ref{sec:games} to get rid of epistemic operators in the goal language.
The difference between private and public information is reminiscent of the work of Christoff and Hansen \cite{ChristoffHansenJAL2015,ChristoffEtAlJOLLI}, which also does not focus on strategic aspects.  A related problem to opinion diffusion is that of information cascades and knowledge diffusion, which has been given formal treatment in a logical settings \cite{RuanThielscherer2011,BaltagEtAl2013}.
Finally, our work is greatly indebted to the work of  \cite{GutierrezEtAlIC20015}, since an influence game can be considered as a variation of an iterated boolean game in which individuals do not have direct power on all the variables -- there can be several individuals influencing another one -- but concurrently participate in its change.
Finally, \cite{SeligmanIBG2015} recently presented an extension
of iterated boolean games with a social network structure in which agents choose actions depending on the actions of those in their neighbourhood.

The paper is organized as follows.
Section~\ref{section:definitions} presents the basic definition of private and public opinions, as well as our model of opinion diffusion.
In Section~\ref{sec:goals} we present our language for goals based on an epistemic version of LTL, and we show that both the model-checking problem remains in PSPACE (as for LTL), by showing a reduction of the epistemic operator.
Section~\ref{sec:games} introduces the definition of influence games, and presents the main results about the effects of the network structure on solution concepts such as Nash equilibria and winning strategy, and on the complexity of checking that a given profile of strategies is a Nash equilibrium.
Section~\ref{sec:conclusions} concludes the paper.

%%%%%%%%%%%%%%%%%%%%%%%%%%%%%%%%%%%%%%%%%%%%%%%%%%%%%%%%
\section{Opinion diffusion }\label{section:definitions}%
%%%%%%%%%%%%%%%%%%%%%%%%%%%%%%%%%%%%%%%%%%%%%%%%%%%%%%%%

In this section we present the model of opinion diffusion which is the starting point of our analysis.
We generalise the model of propositional opinion diffusion introduced in related work \cite{GrandiEtAlAAMAS2015} by separating private and public opinions, and adapting the notion of diffusion through aggregation to this more complex setting.

%%%%%%%%%%%%%%%%%%%%%%%%%%
\subsection{Private and public opinions}

\noindent
Let $\I=\{p_1,\dots,p_m\}$ be a finite set of propositions or \emph{issues}
and let $\N=\{1,\dots,n\}$ be a finite set of individuals or \emph{agents}.
Agents have opinions about all issues in $\I$ in the form of a propositional evaluations, or,  equivalently, a vector of 0s and 1s:

\begin{definition}[private opinion]
The \emph{private opinion} of agent $i$ is a function
%\begin{align*}
 $B_i :  \I \to   \{1,0\}$
%\end{align*}
where
$B_i(p) = 1$ and
$B_i( p) = 0$
%and
%$\B_i(p) = ?$
express, respectively, the agent's opinion that $p$ is true
and the agent's opinion that $p$ is false.
%and
%  agent $i$'s not having an opinion
%about $p$.
\end{definition}

%UG:do we really need this notation?
For every $J \subseteq \N$,
we denote with $\opset_J = \Pi_{i \in J} B_i$
the set of all tuples of
opinions of the agents in $J$.
Elements of $\opset_J $ are denoted by
$\B_J$. For notational convenience,
we write $\opset$
instead of $\opset_\N$, and
  $\opset_i$
instead of $\opset_{  \{i \}  }$.
%and $\B$
%instead of
%$\B_J$.

%additional notation + explanation
Let $\prof{B}=(B_1,\dots,B_n)$ denote the profile composed by the individual opinion of each agent.
Propositional evaluations can be used to represent ballots in a multiple referendum,
expressions of preference over alternatives, or value judgements over correlated issues (see, e.g., \cite{ChristianEtAlRED2007,GrandiEndrissIJCAI2011}).
Depending on the application at hand, an integrity constraint can be introduced to model the propositional correlation among issues.
For the sake of simplicity in this paper we do not assume any correlation among the issues, but the setting can easily be adapted to this more general framework.
%[[should we add an example?]]

We also assume that each agent has the possibility of declaring or hiding
her private opinion on each of the issues. %opinions of agents are either visible or invisible to the other agents. This is decided by each agents by means
%of an additional visibility function.

\begin{definition}[visibility function]
The \emph{visibility function} of agent~$i$ is a map
%\begin{align*}
 $\vis_i: \I \to   \{1,0\}$
%\end{align*}
where
$\vis_i (p) = 1$ and
$\vis_i ( p) = 0$
%and
%$\B_i(p) = ?$
express, respectively, the fact that
agent $i$'s opinion on $p$
is visible and the fact that
agent $i$'s opinion on $p$ is hidden.
\end{definition}

We denote with $\prof{\vis} = (\vis_1, \ldots, \vis_n)$
the profile composed of the agents' visibility functions.
By combining the private opinion
with the visibility function of an agent
we can build her public opinion as a
three-valued function on the set of issues.

\begin{definition}[public opinion]
Let $B_i$ be agent $i$'s opinion
and $\vis_i$ her visibility function.
The \emph{public opinion}
induced by $B_i$
and $\vis_i $
is a function
%\begin{align*}
 $P_i :  \I \to   \{1,0,?\}$
%\end{align*}
such that
$$
P_i (p) =
\begin{cases}
B_i ( p) & \mbox{ if }  \vis_i(p) =1 \\
 ?  & \mbox{ if }  \vis_i(p) =0
\end{cases}
$$
\end{definition}

%do we also need this notation?
For every $J \subseteq \N$,
we denote with
$\pubopset_J = \Pi_{i \in J} \V_i$
the set of all tuples of
public
opinions of the agents in $J$.
Elements of
$\pubopset_J $
are denoted by
$\V_J$. For notational convenience,
we write $\pubopset$
instead of $\pubopset_\N$, and
  $\pubopset_i$
instead of $\pubopset_{  \{i \}  }$.
%and $\V$
%instead of
%$\V_J$.

Once more, $\prof{P}=(P_1,\dots,P_n)$ denotes the profile
of public opinions of all the agents in $\N$.
$P_i$ is aimed at capturing the \emph{public} expression of $i$'s view about the issues in $\I$.
Observe that an agent can only hide or declare her opinion about a
given issue, but is not allowed to lie.
Relaxing this assumption would actually represent an interesting direction for future work.

%commented because already in POF2015
%Examples of public opinions are expressions
%of value judgments
%(.e., whether a state of
%affairs should be considered good or bad), expressions of preferences
%(i.e., whether object $a$ is preferred to object $b$)
%or desires(i.e., whether a state of
%affairs is de sirable or not) and, finally, expressions of choices
%(i.e., whether an action should be performed or not).

%We write
%$\B = (\B_1, \ldots, \B_n)$
%to denote
%the tuple including the opinions  of all agents,
%$\vis = (\vis_1, \ldots, \vis_n)$
%to denote
%the tuple including the visible issues of all agents
%and
%$\V = (\V_1, \ldots, \V_n)$
%to denote
%the tuple including the   public opinions  of all agents.
%We call them, respectively, \emph{opinion profile},
%\emph{visibility profile}
%and \emph{public opinion profile}.
%%The set of all opinion profiles is denoted by $\opset$
%%and the set of all public opinion profiles is denoted by $\pubopset$.

\subsection{Information states}

The information contained in a profile of public opinions can
also be modelled using a state-based representation and an
indistinguishability relation, in line with the existing work on
interpreted systems (see, e.g., \cite{FaginEtAl1995}). States will form the building blocks
of our model of strategic reasoning in opinion dynamics.

\begin{definition}[state]
A \emph{state} is a tuple $  S = (\B, \prof{\vis})$
where $\B $
is a profile of private opinions and
$\prof{\vis}$
is a profile of visibility functions.
The set of all states
is denoted by $\stateset$.

\end{definition}

The following definition formalises the uncertainty between states induced by
the visibility functions.
The idea is that an agent cannot
distinguish between two states
if and only if both
the agent's individual opinion
and the other agents' public opinions
are the same according to the two states.

\begin{definition}[Indistinguishability]\label{def:indistingui}
Let $S, S' \in \stateset$
be two states.
We say that agent $i$
cannot distinguish between
them, denoted by $ S \sim_i S'$,
if and only if:
\begin{itemize}
\item $B_i = B_i'$,
\item $\prof{\vis} = \prof{\vis}'$ and
\item for all $ j \in \N \setminus \{i \}$
and for all $p \in \I$,
if  $ \vis_j(p) = 1$ then $B_j (p) = B_j' (p)$.
\end{itemize}
\end{definition}

Let
%\begin{align*}
$S^{\sim_i} =  \{ S' \in  \stateset \mid S \sim_i S'  \}$
%\end{align*}
be the set of states that agent $i$ cannot distinguish
from $S$. Clearly $\sim_i$ is an equivalence relation.
In what follows we will often use public states to represent the equivalence class of a state. Observe however that this is a too coarse representation, since each of the agents knows her own belief.

\begin{example}\label{ex:equiv}
Let there be three agents $i,j$ and $k$, and one issue $p$.
Assume that agents $j$ and $k$ consider $p$ as true as private opinion while $i$ private opinion is $p=0$.
Suppose also that agents $i$ and $j$ make their opinion public while agent $k$ does not.
From the perspective of agent $i$, the two states $S_0 = ((0,1,0), (1,1, 0))$ and  $S_1 = ((0,1,1), (1,1, 0))$ are indistinguishable, representing the two possible opinions of agent $k$ who is hiding it. Figure \ref{fig:equiv} represents the perspective of agent $i$.
\end{example}

\begin{figure}[h]
\begin{center}
\begin{tikzpicture}
\node (s1) at (0,0) [circle,draw,font=\scriptsize,
                                         label={[align=left] $B_i(p) = 0$\\$B_j(p) = 1$},  label={below:$ B_k(p) = 0$}] {$S_0$};
\node (s2) at (3,0) [circle,draw,font=\scriptsize,
                                         label={[align=left]$B_i(p) = 0$\\$B_j(p) = 1$},  label={below: $B_k(p) = 1$}]{$S_1$};
\draw [-]  (s1) -- (s2) node[draw=none,fill=none,midway,below] {$i$} ;
\end{tikzpicture}
\end{center}
\caption{Indistinguishable states for  $i$ as $k$ hides $p$.}
\label{fig:equiv}
\end{figure}
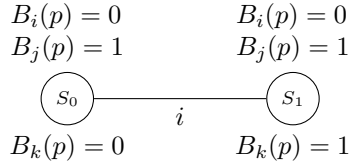

\subsection{Opinion diffusion
through aggregation}
%%%%%%%%%%%%%%%%%%%%%%%%%%%%%%%%%%%%%%%%%%%%%%%%%%

\noindent
In this section we define the influence process that is at the heart of our model.
Our definition is a generalisation of the model by \cite{GrandiEtAlAAMAS2015}.

First, we assume that individuals are connected by an \emph{influence network}
which we model as a directed graph:

\begin{definition}[influence network]
An \emph{influence network}
is a directed irreflexive  graph $E\subseteq \N\times\N$.
We interpret $(i,j)\in E$ as ``agent $j$ is influenced by agent $i$''.
\end{definition}

We also refer to $E$ as the influence graph and to individuals in $\N$ as the nodes of the graph.
%Note that the influence network is directed, hence $(i,j)\in E$ represents the fact that $i$ influences $j$.
Let $\Inf(j) = \{i\in N \mid (i,j)\in E \}$ be the set of \emph{influencers} of agent $j$ in the network $E$.
Given a state $S$, this definition can be refined by considering the set
$\Inf^{S}(i,p) = \{j\in N \mid (j,i)\in E \text{ and } P_j(p)\not=\;?  \}$ to be the subset of influencers that are actually expressing their private opinion about issue $p$.
Clearly, $\Inf^{S}(i,p)\subseteq \Inf(i)$ for all $p$ and $S$.
%We omit the reference to the state $S$ if clear from the context.

\begin{example}
Figure \ref{fig:in} represents a basic influence network where some agent $i$ is influenced by two agents $j$ and $k$. The set $\Inf(i)=\{j,k\}$, and, using the notation in the previous example, the set $\Inf^{S_0}(i,p)=\Inf^{S_1}(i,p)=\{j\}$.

\begin{figure}[h]

\begin{center}
\begin{tikzpicture}
\node (ai) at (0,0) {$i$};
\node (aj) at (1,0) {$k$};
\node (ak) at (-1,0) {$j$};
%\node (al) at (1,0) {$l$};
\draw [->]  (aj) -- (ai) ;
\draw [->]  (ak) -- (ai) ;
%\draw [->]  (ai) -- (al) ;
\end{tikzpicture}
\end{center}
\caption{$i$ influences by $j$ and $k$.}
%: $Inf(i) = \{j, k\}$}
\label{fig:in}
\end{figure}

\end{example}

Given a profile of public opinions and % a confidence
an influence network $E$, we model the process of opinion diffusion by means of an aggregation function, which shapes the private opinion of an agent by taking into consideration
the public opinions of her influencers.

%An  aggregation procedure $F_i$   computes how an agent $i$ changes her individual opinions starting from the public opinions of her influencers.
%Given an integrity constraint $\ic$ that defines a set of feasible opinions $\X$,
%In particular;

\begin{definition}[Aggregation procedure]\label{def:aggregation}
An \emph{aggregation procedure} for agent $i$ is a class of functions
\begin{align*}
% F_i: 2^\stateset \times 2^\N \longrightarrow \opset \text{ for each } J \subseteq \N \setminus \{    i \}
 F_i: \opset \times \pubopset_J \longrightarrow \opset \text{ for each } J \subseteq \N \setminus \{    i \}
 \end{align*}
that maps agent $i$'s individual opinion and the public opinions of a set of agents $J$ to agent $i$'s individual opinion.
%that maps an agent $i$'s epistemic state, together with a set of individuals (its influencers) to a private opinion for agent $i$.
\end{definition}

Aggregation functions are used to construct the new private opinion of an agent in the dynamic process of opinion diffusion. Thus, $F_i(B_i,\prof{P}_{\Inf(i)})$ represents the private opinion of agent $i$ updated with the public opinions received by its influencers.
%\footnote{[[UG:give examples, and talk about collective rationality in case there are integrity constraints. I think we cannot consider integrity constraints, it will mess us our proofs in the logic section. LP: does it mean that the influence is irreflexive as aggregation always take into account initial private opinion?]]}

A number of aggregation procedures have been considered in the literature on judgment aggregation and can be adapted to our setting.
Notable examples are quota rules,
where an agent changes her opinion if the amount of people disagreeing with her is superior of a given quota (the majority rule is such an example).
These aggregation procedures give rise to the class of threshold models studied in the literature on opinion diffusion \cite{Granovetter1978,SchellingMicromotives}.
%Quota rules gives rise to so-called threshold models [[CITE]], and the dynamics generated by these and other aggregation procedures have been studied in related work \cite{GrandiEtAlAAMAS2015}.

For the sake of simplicity in this paper we  consider that all agents use the following
aggregation procedure:

\begin{definition}\label{def:agg}
Let $S=(\B,\prof{V})$ be a state and $\prof{P}$ the corresponding profile of public opinions.
The \emph{unanimous issue-by-issue aggregation procedure} is defined as follows:

%$$ F_i(B_i,\prof{P}_{\Inf(i)})(p) = \begin{cases}
%B_i(p) & \mbox{ if } \Inf(i)=\emptyset\\
%x & \mbox{ if }  P_j(p)=x \mbox{ for all } j\in\Inf(i) \\
%{} & \text{ such that } P_j\not = \;?\\
%B_i(p) & \mbox{otherwise}
%\end{cases}$$
%\end{definition}

$$ F^U_i(B_i,\prof{P})(p) = \begin{cases}
B_i(p) & \mbox{ if } \Inf^{S}(i,p)=\emptyset\\
x\in \{0,1\} & \mbox{ if }  P_j(p)=x \;\;\forall j\in\Inf^{S}(i,p) \\
B_i(p) & \mbox{otherwise}
\end{cases}$$
\end{definition}

That is, an individual will change her private opinion about issue $p$ if and only if
all her influencers that are expressing their opinion publicly are unanimous in disagreeing with her own one.

%%%%%%%%
\subsection{Strategic actions and state transitions}

Showing or hiding information is a key action in the model of opinion diffusion defined above. The dynamic of opinion is %hence
 rooted in two dimensions: the influence network and the visibility function. At each time step, by hiding or revealing their opinions, agents influence other agents opinions.
We assume that agents can make their opinions visible or invisible
by specific actions of type
$\reveal{p}$
(i.e.,
action of making  the opinion about $p$ visible)
and
$\hide{p}$
(i.e.,
action of hiding the opinion about $p$).
The action of doing nothing is denoted by $\skipact$.
Let therefore
\begin{align*}
\actions = & \{  \reveal{p} : p \in \I  \} \cup \{  \hide{p} : p \in \I  \}
 \cup \{  \skipact \}
\end{align*}
be the set of all individual actions
and
$\jactions = \actions^n$
the set of all joint actions.
Elements of $\jactions$
are denoted by $\prof{a}=(a_1,\dots,a_n)$.
%As usual,  $\delta_i$
%denotes the element in $\delta
%$ corresponding to agent $i$.
%perch� non le chiamiamo a?

%\begin{definition}
%A joint action sequence is
%a function
%\begin{align*}
%\lambda : \mathbb{N} \longrightarrow \jactions
%\end{align*}
%The set of all joint action sequences is denoted by $\actseq$.
%\end{definition}

Each joint action $\prof{a}$ induces a transition function between states.
This function is deterministic and is defined as follows:
%The successor functionis the function
%describing how joint action deterministically changes
%a state. In particular:

\begin{definition}[transition function]
The \emph{transition function}
%\begin{align*}
$\succfunct: \stateset \times \jactions \longrightarrow \stateset$
%\end{align*}
associates to each state $S$ and joint action $\prof{a}$ a new state
$S'= (\prof{B}',\prof{\vis}') $
where, for all $i \in \N$:
\begin{itemize}
\item $
 \vis'_i (p)=
 \begin{cases}
 1 & \text{ if }  a_i=  \reveal{p}    \text{ or } \\
 %& ( a_i = \skipact  \text{ and } \vis_i (p) = 1)\\
 0 & \text{ if } a_i=  \hide{p}    \text{ or } \\
 %& ( a_i = \skipact  \text{ and } \vis_i (p) = 0)
 V_i(p) & \text{ if } a_i=\skipact
\end{cases}
$

\item $B_i'  =   F_i^U(B_i, \prof{P}'_{\Inf(i)})$

\end{itemize}
\noindent
Where $\prof{P'}$ is the public profile obtained from private profile $\prof{B}$ and visibility functions~$\prof{V'}$.
\end{definition}

By a slight abuse of notation we denote with $\prof{a}(S)$ the state $\succfunct(S,\prof{a})$ obtained from $S$ and $\prof{a}$ by applying the transition function.
Observe that in our definition of transition function we are assuming that
the influence process occurs after that
the actions
have modified the visibility
of the agents' opinions. Specifically,
first, actions   have consequences
on the visibility
of the agents' opinions, then,
each agent modifies her private opinions
on the basis
of those opinions of other agents
that have become public.

%\begin{definition}[Induced history]
%Let $S_0 = (\B_0, \vis_0) $
%be an initial state
%and let $\lambda$ be a joint action sequence.
%The history $H_{S_0, \lambda} $ determined by them
%is defined
%as follows:
%\begin{align*}
%H_{S_0, \lambda} (0) & =  S_0\\
%H_{S_0, \lambda} (n+1) & =  \succfunct( S_n, \lambda(n)        ) \text{ for all } n \in  \mathbb{N}
%\end{align*}
%\end{definition}

%The previous section described how individual opinions evolve over time,
%and how agents have a simple set of action at their disposal to influence
%this process.
%Let us now give a broader perspective to opinion dynamics by considering a sequence of revealing and hiding actions.
We are now ready to define the concept of \emph{history}, describing the temporal aspect
of agents' opinion dynamic:

\begin{definition}[history]
Given a set of issues $\I$, a set of agents $\N$, and aggregation procedures $F_i$ over
a network $E$,  an \emph{history} is an infinite sequence of states
$H = (H_0, H_1, \ldots)$.
such that for all $t\in\mathbb N$ there exists a joint
action $\prof{a}_t\in \mathcal J$ such that $H_{t+1}=\prof{a}_t(H_n)$.
\end{definition}

Let $H = (H_0, H_1, \ldots)$ be an history.
For notational convenience,
for any $i \in \N$ and
for any $t \in \mathbb{N}$,
we denote with $H_{i,t}^B$ agent $i$'s private opinion in state $H_t$
and with $H_{i,t}^\vis$ agent $i$'s visibility function in state $H_t$.

The set of all histories is denoted by $\histset$. Observe that our definition restricts the
set of all possible histories to those that corresponds to a run of the influence dynamic described above.

\begin{example}
Let us reconsider the two previous examples,
%hereafter, we only focus on agents $i$, $j$ and $k$.
with initial state $H_0 = S_0$.
% and $j$ and $k$ are agent $i$ influencers.
Consider now the following joint actions $\prof{a}_0 = (\skipact, \skipact,  \reveal{p})$ and $\prof{a}_1 = (\skipact,  \hide{p}, \skipact)$: agent $k$ reveals her opinion, and at the next step $j$ hides her opinion about $p$.
If we assume that all individuals are using the unanimous aggregation procedure %(see Definition~\ref{def:agg})
then Figure \ref{fig:history} shows the two states $H_1$ and $H_2$ constructed by applying the two joint actions from state $S_0$.
In state $H_1$, agent $i$'s private opinion about $p$ has changed, i.e., $H_{i,1}^B(p)=1$ as all her influencers are publicly unanimous about $p$.
At the next step, instead, no opinion is updated.
%Next,  $H_{i,2}^B$ does not change, as her influencers are not publicly unanimous about p.
%UG:no, they are, it's just that one agent is hiding.
\end{example}

\begin{figure}[h]
\begin{center}
\begin{tikzpicture}
\node (h0) at (0,0) [font=\scriptsize,label={below:$H_0$}] {$((0,1,1), (1,1, 0))$} ;
\node (h1) at (3,0) [font=\scriptsize,label={below:$H_1$}] {$((1,1,1), (1,1, 1))$} ;
\node (h2) at (6,0) [font=\scriptsize,label={below:$H_2$}] {$((1,1,1), (1,0, 1))$} ;
\draw [->]  (h0) -- (h1) node[draw=none,fill=none,midway,below] {$\prof{a}_0$} ;
\draw [->]  (h1) -- (h2) node[draw=none,fill=none,midway,below] {$\prof{a}_1$} ;
\end{tikzpicture}
\end{center}
\caption{The initial two states of a history.}
%First states of History $H$ entailed by $\prof{a}_0$ and $\prof{a}_1$}
\label{fig:history}
\end{figure}

%%%%%%%%%%%%%%%%%%%%%%%%%%%%%%%%%%%%%%%%%%%%%%%%%%%%%%%%%%%%%%%%%
\section{Temporal and Epistemic Goals}\label{sec:goals}     %%
%%%%%%%%%%%%%%%%%%%%%%%%%%%%%%%%%%%%%%%%%%%%%%%%%%%%%%%%%%%%%%%%%

As agents can hide or reveal their opinions, the strategic dimension of opinion diffusion is immediate: by revealing/hiding her opinion, an agent influences other agents.
Influence is actually guided by some underlying goals.
An agent reveals/hides her opinion only if she wants to influence some other agents: she aims at changing individual opinions.
Temporal dimension is immediate as the dynamics of opinion has to be considered. Following \cite{GutierrezEtAlIC20015},  we define a language to express individual epistemic goals about the state of the individual opinions.

%%%%%%%%%%%%%%%%%%%%%%%%%%%%%%%%%%%%%%%%%%%%%%%%
\subsection{Epistemic temporal logic of influence}

Let us introduce a logical language based on
a combination of simple version of multi-agent epistemic logic and
linear temporal logic ($\mathsf{LTL}$)
that can be interpreted over histories. In line with our framework, term \emph{epistemic state} should be  interpreted as \emph{private opinion}. Goals in our perspective consists of targeting an epistemic state: typically `'agent $i$ wants that agent $j$ has private opinion $\varphi$ in the future''. The proposed language does not allow the temporal operator to be in the scope of a knowledge operator, obtaining a simpler language and a reduction that allows us to stay in the same complexity class as $\mathsf{LTL}$.

We call $\logic$ this logic, from epistemic linear temporal logic of influence.
Its language, denoted by $\lang$, is defined by the following
BNF:

\begin{center}\begin{tabular}{lcl}
 $\alpha$  & $\bnf$ & $\bel{i}{p}  \mid \visatom{i}{p} \mid \neg\alpha \mid \alpha_1 \wedge \alpha_2 \mid \know{i}{\alpha} $\\
 $\varphi$  & $\bnf$ & $\alpha  \mid \neg\varphi \mid \varphi_1 \wedge \varphi_2  \mid \nextop{\varphi} \mid \until{\varphi_1}{\varphi_2} $\
\end{tabular}\end{center}
where $i$ ranges over $\N$
and $p$
ranges over $\I$.
$\bel{i}{p}$ has to be read
``agent $i$'s opinion is that $p$
is true'' while %, clearly,
$\neg \bel{i}{p}$ has to be read
``agent $i$'s opinion is that $p$
is not true'' (since we assume that agents have binary opinions).
$\visatom{i}{p}$ has to be read
``agent $i$'s opinion about $p$ is visible''.
Finally, $ \know{i}{\alpha}$
has to be read
``agent $i$ knows that $\alpha$ is true''.

$\nextop{\varphi}$
and $\until{ }{ } $
are the standard $\mathsf{LTL}$ operators
`next' and `until'.
In particular,
$ \nextop{\varphi}$
has to be read
``$\varphi$ is going to be true in the next state''
and
$\until{\varphi_1}{\varphi_2}$
has to be read
``$\varphi_1$ will be true until $\varphi_2$ is true''.
As usual,
we can define the temporal
operators
`henceforth' ($\henceforth{ }$) and `eventually'
($\eventually{ }$)
by means of the
`until' operator:
\begin{align*}
\henceforth{\varphi} & =_{\mathit{def}} \neg (\until{\top}{\neg \varphi})  \\
\eventually{\varphi}   & =_{\mathit{def}} \neg \henceforth{\neg \varphi }
\end{align*}

The interpretation
of $\lang$-formulas
relative to histories is  defined as follows.
\begin{definition}[Truth conditions]
Let $\varphi$ be
a $\lang$-formula, let $H  $ be a history and let $k \in \mathbb{N}$. Then:
\begin{eqnarray*}
H, k \models \bel{i}{p} & \Leftrightarrow & H_{i,k}^B (p) =1   \\
H, k \models \visatom{i}{p} & \Leftrightarrow & H_{i,k}^\vis (p) =1   \\
H, k \models \know{i}{\alpha} & \Leftrightarrow & \forall H' \in \histset : \text{ if } H(k) \sim_i H' (k)  \text{ then }\\
&& H', k \models \alpha\\
H, k \models \neg \varphi & \Leftrightarrow & H, k \not \models  \varphi   \\
H, k \models  \varphi_1 \wedge \varphi_2 & \Leftrightarrow & H, k   \models  \varphi_1 \text{ and }  H, k  \models  \varphi_2  \\
H, k \models  \nextop{\varphi} & \Leftrightarrow & H, k+1 \models  \varphi    \\
H, k \models  \until{\varphi_1}{\varphi_2}  & \Leftrightarrow & \exists k' \in \mathbb{N}  :   ( k \leq k' \text{ and }  H, k' \models  \varphi_2 \text{ and }  \\
& & \forall k'' \in \mathbb{N}: \text{ if } k \leq k'' < k' \text{ then } H, k''  \models  \varphi_1 )
\end{eqnarray*}

\end{definition}

The operator $\know{i}{}$ is rather peculiar, and should not be interpreted
as a classical individual epistemic operator. It mixes public and private opinions of our model.  Operator  $\know{i}{}$  reading is rather  "agent $i$ is uncertain about other agents private opinion as this opinion is not visible" and $\know{i}{\alpha}$ stands for agent $i$ knows $\alpha$ despite this uncertainty.

The following proposition shows that $\know{i}{}$ could also be formulated
in terms of equivalence between histories rather than in terms
of equivalence between states. Its proof is immediate from our definitions.

\begin{proposition}
Let $ H\in \histset $
and $ k \in \mathbb{N}$.
Then
\begin{eqnarray*}
H, k \models \know{i}{\alpha} &  \text{ iff } & \forall H' \in \histset : \text{ if } H \sim_i H'  \text{ then }
H', k \models \alpha
\end{eqnarray*}
where $H \sim_i H'$ iff $H(h) \sim_i H'(h)$
for all $h \in \mathbb{N}$.
\end{proposition}

%\begin{proof}[Sketch]
%[[TODO]]
%\end{proof}

%The following two propositions clarify this statement.
\begin{example}
Consider Figure \ref{fig:history}, the following statement expresses that in state $H_0$, it is the case that in the future agent $k$ knows agent $i$ public opinion about $p$ (as $H_1^{\sim_k} = \{H_1\}$):
\[
H, 0 \models \eventually{(\know{k}{(\bel{i}{p}\land}\visatom{i}{p}))}
\]
\end{example}
This example also shows how each $\logic$ statements can be used for representing individual goals. Hence, gathering individual goals lead to the construction of a boolean game. Before detailing this aspect we conclude the section by exhibiting the key results about model checking % on the complexity of
for the $\logic$ logic.

%%%%%%%%%%%%%%%%%%%%%%%%%%%%%%%%%%%%%%%%%%%%%%%%%%%%%
\subsection{Model checking}\label{sec:reduction}

The aim of this section is to show that model checking for $\logic$ is as hard as model checking for $\mathsf{LTL}$.
Recall that epistemic temporal logic has very high complexity \cite{FaginEtAl1995}.
%, showing that the addition of epistemic goals does not change the computational complexity of the problem.
We do so by reducing formulas containing an epistemic modality to propositional ones.
%[[UG: what is the complexity of $\mathsf{ELTL}$ in its original formulation, mixing temporal and epistemic operators?]]

\begin{lemma}\label{lemma:val1}
The following formulas are valid in $\logic$:
\begin{enumerate}[nolistsep,label=$(\roman*)$]
\item   $\know{i}{\bel{i}{p}}  \leftrightarrow  \bel{i}{p}$
\item   $\know{i}{\bel{j}{p}}  \leftrightarrow  (\bel{j}{p}  \wedge \visatom{j}{p} ) \text{ if }  i \neq j$
\item   $\know{i}{\visatom{j}{p}}  \leftrightarrow  \visatom{j}{p}$ for all $j$
\item   $\know{i}{\neg\bel{i}{p}}  \leftrightarrow  \neg\bel{i}{p}$
\item   $\know{i}{\neg\bel{j}{p}}  \leftrightarrow  (\neg\bel{j}{p}  \wedge \visatom{j}{p} ) \text{ if }  i \neq j$
\item   $\know{i}{\neg\visatom{j}{p}}  \leftrightarrow  \neg\visatom{j}{p}$ for all $j$

\end{enumerate}
\end{lemma}

\begin{proof}[sketch]
Straightforward from Definition~\ref{def:indistingui} and the interpretation of the $\know{i}{p}$ operator. To show the right-to-left direction of $(ii)$ and $(v)$, suppose that $\bel{j}{p}$ is true at every indistinguishable state for $\sim_i$, i.e., that $\know{i}{\bel{j}{p}}$ is true. This %clearly
implies that $\bel{j}{p}$ is true in the current state. Moreover, if $\visatom{j}{p}$ is false, then by Definition~\ref{def:indistingui} there would be an indistinguishable state in which $\bel{j}{p}$ is false, contradicting the hypothesis.
\end{proof}

\noindent
We now show two distribution laws for the operator $\know{i}{}$ with respect to conjunction and negation:

\begin{lemma}\label{lemma:val2}
The following formulas are valid in $\logic$:
\begin{enumerate}[nolistsep,label=$(\roman*)$]
\item $\know{i}{(\alpha_1\wedge\alpha_2)}\leftrightarrow ( \know{i}{\alpha_1} \wedge \know{i}{\alpha_2})$ for all $\alpha_1$ and $\alpha_2$;
\item $\know{i}{(\alpha_1\vee \alpha_2)}\leftrightarrow  (\know{i}{\alpha_1} \vee \know{i}{\alpha_2})$ when $\alpha_1$ and $\alpha_2$ do not contain any occurrence of the modality $\know{j}{}$ for any $j$.
\end{enumerate}
\end{lemma}

\begin{proof}[sketch]
$(i)$ and the right-to-left directions of $(ii)$ are standard consequences of interpreting $\know{i}{p}$ over equivalence relations. We now prove the left-to-right direction of $(ii)$ by induction on the construction of a propositional formula.

Assume some history $H$ and state $S$ and suppose that the left part holds.
This means that in all states $\in S^{\sim_i}$, either $\alpha_1$ or $\alpha_2$ is true.
If both statements contains $\visatom{j}{p}$ or  $\neg\visatom{j}{p}$ then  Definition \ref{def:indistingui} entails that right-to-left direction hold, since either $\alpha_1$  holds in all  states $\in S^{\sim_i}$ or $\alpha_2$ holds in all  states $\in S^{\sim_i}$.
If $\alpha_1$ (respectively $\alpha_2$) contains some statement $\bel{i}{p}$, then the right-to-left direction holds as $\alpha_1$  either holds in all states $\in S^{\sim_i}$ or is false in all states.
If it does not hold then $\alpha_2$ is considered in a similar way to $\alpha_1$.
Now suppose that $\alpha_1$ and  $\alpha_2$  only contain  statements of the form $\bel{j}{p}$ s.t. it is always the case that $j\ne i$.
Then there must exist $j, p$ such that $\visatom{j}{p}= 1$. Otherwise, there exists a state $S'\in  S^{\sim_i}$ such that neither $\alpha_1$ or $\alpha_2$ hold (as all possible indistinguishable states must be considered).
In conclusion, the right-to-left direction holds as either $\alpha_1$  (respectively $\alpha_2$) either holds in all states $\in S^{\sim_i}$ or it does not hold in all states.
%Suppose some state $S'\in S^{\sim_i} = (\B', \prof{\vis}')$ such that $H, S'\models \alpha_1$. If $\alpha_1 = \bel{i}{p}$ then definition of indistinguishability entails that left-to-right direction hold. Now suppose$\alpha_1 = \bel{j}{p}$ $(j \ne i)$, in that case, left-to-right direction is again  proved if $ \vis_j'(p) = 1$ (again def. of indistiguishability). Suppose that $\vis_j'(p) = 0$, then it means that there exists $S''\in $ and $S'$ and $S''$ represent to the two possible indistinguishable states about $$j'$ private opinion. Hence,  $\alpha_2$ holds in state $S''$ as $\alpha_1$ does not hold.
%The key stp is when $\alpha_1=\bel{j}{p}$ and $\alpha_2=state $S'\in S^{\sim_i} = (\B', \prof{\vis}')$ such that $H, S'\models \alpha_1$. If $\alpha_1 = \bel{i}{p}$ then definition of indistinguishability entails that left-to-right direction hold. $.
%In this simple case with only two issues, the set of states expressed in terms of public information is $\{01,00,11,10,?0,?1,?0, ?1, ??\}$. Now $\alpha_1\vee \alpha_2$ is true at $\{01,11,10,?1,1?\}$.
%In all of these public states the formula $\know{i}{\alpha_1}$ or the second formula $\know{i}{alpha_2}$ is true. \footnote{[Umberto:] I can't finish this proof! There should be an induction or so...}
\end{proof}

\noindent
Finally, we reduce the nesting of the $\know{i}{p}$ operator:

\begin{lemma}\label{lemma:depth}
The following formulas are valid in $\logic$:
\begin{enumerate}[nolistsep,label=$(\roman*)$]
\item $\know{i}{\know{i}{\alpha}} \leftrightarrow \know{i}{\alpha}$
\item $\know{i}{\know{j}{\know{i}{\alpha}}} \leftrightarrow \know{j}{\know{i}{\alpha}$} for all $i\not= j$, when $\alpha$ do not contain any occurrence of the modality $\know{j}{}$ for any $j$.
\end{enumerate}
\end{lemma}

\begin{proof}[sketch]
The proof of $(i)$ is a standard consequence of using equivalence relations.
To prove $(ii)$, let $\alpha$ be a propositional formula.
Put first $\alpha$ in CNF and then distribute by Lemma~\ref{lemma:val2} the modalities over conjunction.
We now prove that $\know{i}{\know{j}{\know{i}{\ell}}} \leftrightarrow \know{j}{\know{i}{\ell}}$ where $\ell$ is a literal. If $\ell=(\neg)\visatom{j}{p}$, then by Lemma~\ref{lemma:val1} both sides of the equivalence reduce to $(\neg)\visatom{j}{p}$.
Suppose that $\ell=(\neg)\bel{k}{p}$ for some $k\in\N$.
If $k\not=i,j$ then by Lemma~\ref{lemma:val1} we can reduce $\know{i}{\ell}$ to $\ell\wedge \visatom{k}{p}$, and then it is straightforward to conclude by observing that $\know{i}{\ell\wedge \visatom{k}{p}}\leftrightarrow \ell\wedge \visatom{k}{p} \wedge \visatom{k}{p}$ which in turn is equivalent to $\ell\wedge \visatom{k}{p}$.
The case of $k=i$ and $k=j$ is similar.
\end{proof}

\noindent
We are now ready to present an algorithm to translate a formula of $\logic$ into an equivalent one
without any occurrence of the $\know{i}{}$ operator. Let an epistemic literal be a formula of the form $(\neg)\know{i}{\ell}$ for $i\in\N$ and propositional literal $\ell$.

%%%%%%%%%%%%%%%%%%%%%%%

%\el{In the algorithm we do not need to impose that $\varphi$ has no temporal operators. I have removed this constraint to make it more general. This allows us to have a   reduction from $\logic$ to $\logicminus $.}
\begin{algorithm}

 \KwIn{a formula $\varphi\in \lang$ }
 \KwOut{a formula $\redu(\varphi)\in \lang$ with no epistemic operators   }

\smallskip

	\While{there is an epistemic operator in $\varphi$ outside an epistemic literal}{
1. choose a subformula $\know{i}{\alpha}$ of $\varphi$
   such that $\alpha$ is without epistemic operators and is not a propositional literal\;
2. put $\alpha$ in negated normal form (NNF)\;
3. distribute $\know{i}{}$ over $\wedge$ and $\vee$ \;
}
4. Reduce the depth modalities with Lemma~\ref{lemma:depth}\;
5. Reduce the atoms with Lemma~\ref{lemma:val1} \;
\medskip
\caption{Reduction of the epistemic operator}
\label{algo:redproc}
 \end{algorithm}
%%%%%%%%%%%%%%

The following proposition guarantees that the translation defined is also polynomial.

\begin{lemma}\label{redprocedure}
Algorithm \ref{algo:redproc} terminates, $red(\varphi)$ is polynomial in the size of $\varphi$,
does not contain epistemic operators, and is equivalent to $\varphi $.
\end{lemma}

\begin{proof}
Lines 3, 4 and 5 apply equivalences that are valid by Lemmas~\ref{lemma:val1}, \ref{lemma:val2} and \ref{lemma:depth}
and the fact that the following rule of replacement of equivalents
is admissible in $\logic$:
%\el{Added this. The rule of replacement of equivalents is needed to ensure
%that $ \redu(\varphi)$ is indeed equivalent to $\varphi$.}
\begin{align*}
 &     \frac{   \psi_1 \leftrightarrow \psi_2 }{   \varphi \leftrightarrow \varphi[\psi_1 / \psi_2]    }
\end{align*}

The negated normal form of a propositional formula (treating epistemic literals as propositional literals) is constructed by propositional equivalence and is polynomial in the size of the initial formula.
Finally, since the modal depth is limited by Lemma~\ref{lemma:depth} to $|N|$, we add a maximum of $|N|$ extra variables of the form $\visatom{j}{p}$ for some $j$ and some $p$ to the translation of each epistemic literal.
\end{proof}

Using Lemma \ref{redprocedure},
we are able
to polynomially reduce every formula
of the $\logic$
to an equivalent formula
of the  fragment $\logicminus$
whose language $\langminus$ is defined by the following BNF:
\begin{center}\begin{tabular}{lcl}
 $\varphi$  & $\bnf$ & $
 \bel{i}{p}  \mid \visatom{i}{p} \mid \neg\varphi \mid \varphi_1 \wedge \varphi_2  \mid \nextop{\varphi} \mid \until{\varphi_1}{\varphi_2} $\
\end{tabular}\end{center}
where $i$ ranges over $\N$
and $p$
ranges over $\I$. We first show the following:

\begin{proposition}\label{modelchecking}
The model checking problem of $\logicminus$
is PSPACE-complete.
\end{proposition}

\begin{proof}
To verify membership
it is sufficient to note
that $\logicminus$
is a special instance
of $\mathsf{LTL}$
built out of
the finite set of atomic propositions $\{    \bel{i}{p} :  i \in \N  \text{ and } p \in \I      \}
\cup \{    \visatom{i}{p} :  i \in \N  \text{ and } p \in \I      \}$
and interpreted over a subset of the set
of all possible histories for this language.
Since the model checking problem
for $\mathsf{LTL}$ is
in PSPACE
\cite{SistlaClarkeLTL},
the model checking problem
for  $\logicminus$
should also be
in PSPACE.

To check that model checking
of
$\logicminus$
is PSPACE-hard
we are going to
consider the following fragment
of $\langminus$:
\begin{center}\begin{tabular}{lcl}
 $\varphi$  & $\bnf$ & $
\visatom{i}{p} \mid \neg\varphi \mid \varphi_1 \wedge \varphi_2  \mid \nextop{\varphi} \mid \until{\varphi_1}{\varphi_2} $
\end{tabular}\end{center}
where $i$ ranges over $\N$
and $p$
ranges over $\I$.

It is straightforward to check that
 the
set of histories $\histset$
includes all possible interpretations
for this language which is nothing but
the $\mathsf{LTL}$ language
built out of
the finite set of atomic propositions $\{    \visatom{i}{p} :  i \in \N  \text{ and } p \in \I      \}$.
Since the model checking for $\mathsf{LTL}$ is
is known to be PSPACE-hard
\cite{SistlaClarkeLTL},
it follows that the model checking
for $\logicminus $
is PSPACE-hard too.
\end{proof}

We can now state the following theorem
about complexity of model checking
for $\logic$.

\begin{theorem}\label{thm:completeness}
The model checking problem of $\logic$
is PSPACE-complete.
\end{theorem}

\begin{proof}
By Proposition~\ref{modelchecking} we know that model checking for $\logicminus$ is PSPACE-complete.
Every formula of $\logic$ can be reduced to an equivalent
formula of polynomial size in $\logicminus$ by Lemma~\ref{redprocedure},
showing membership in PSPACE of model checking for $\logic$.
Since $\logicminus$ is a sublogic of $\logic$ we also obtain PSPACE-hardness.
\end{proof}

\section{Games of influence}\label{sec:games}       %%
%%%%%%%%%%%%%%%%%%%%%%%%%%%%%%%%%%%%%%%%%%%%%%%%%%%%%%%%

We are now ready to put together all the definitions introduced in the previous sections and give the following definition:

\begin{definition}[Influence game]

An influence game
is a tuple
$ \infgame = (\N, \I, E, F_i, S_0, \gamma_1, \ldots, \gamma_n )$
where $\N$,
$\I$, $E$ and $S_0$ are, respectively,
a set of agents, a set of issues, an influence network, and
an initial state, $F_i$ are aggregation procedures, one for each agent,
and $\gamma_i \in  \lang$ is agent $i$'s goal.
\end{definition}

%[[Do we really need to define utility?]]
%From an influence game we can define the
%induced utility function over histories.
%
%\begin{definition}[Utility over histories]
%Let
%$ \infgame = (\N, \I, E, S_0,  \gamma_1, \ldots, \gamma_n)$
%be an influence game.
%%over histories
%is the function
%$U_i^\histset :  \histset_{S_0} \longrightarrow \{ 0,1 \}  $
%such that,
%for all $H \in \histset_{S_0}$:
%\begin{align*}
%U_i^\histset ( H ) = 1 & \text{ iff }  H, 0 \models \gamma_i
%\end{align*}
%where
%$\histset_{S_0}$
%is the set of histories
%whose initial element is $S_0$.
%\end{definition}

%%%%%%%%%%%%%%%%%%%%%%%%%%%%%%%%%%%%%%%%%%%%%%%%%%%%%%%%%%%%%%%
\subsection{Strategies}

The following definition introduces
the concept of strategy.
The standard definition would call for a function
that assigns in each point in time an action to each player.
We choose to study simpler state-based strategies:

\begin{definition}[strategy]
A strategy for player $i$  is a function that associates an action to every information state, i.e.,
% \begin{align*}
$\strat_i : \stateset \to \actions$ such that $\strat_i(S)=\strat_i(S')$ whenever $S\sim_i S'$.
%\end{align*}
A strategy profile is a tuple $ \prof{\strat} = (\strat_1, \ldots, \strat_n)$.
%UG: I don't think we need this notation
%The set of all strategy profiles is denoted by $\stratset$.
\end{definition}

%That is, a strategy maps a profile of public opinions to a profile of actions.
%This is achieved by imposing that the same action should be taken in indistinguishable
%states.

For notational convenience, we interchangeably use
$\prof{\strat}$ to denote a strategy profile $(\strat_1, \ldots, \strat_n)$
and the function $\prof{\strat} : \stateset \longrightarrow \jactions$
such that $\prof{\strat} (S ) = \prof{a}$ if and only if $\strat_i (S ) = a_i$, for all $S \in \stateset$
 and  $i \in \N$.
As the following definition highlights, every strategy profile $\prof{\strat}$ combined with an initial state $ S_0$ induces a history:

\begin{definition}[Induced history]\label{def:inducedhistory}
Let $S_0$ be an initial state
and let
$ \prof{\strat}$
be a strategy profile.
The history $H_{S_0, \prof{\strat}} \in \histset$ induced by them
is defined
as follows:
\begin{align*}
H_0 (S_0, \prof{\strat}) & =  S_0\\
H_{n+1} (S_0, \prof{\strat}) & = \succfunct(S_n, \prof{\strat}(S_n))  \text{ for all } n \in  \mathbb{N}
\end{align*}
\end{definition}

%The following definition
%introduces the concept
%of utility over strategy profiles.
%
%
%\begin{definition}[Utility over strategy profiles]
%Let
%$ \infgame = (\N, \I, E, S_0,  \gamma_1, \ldots, \gamma_n)$
%be an influence game.
%Agent $i$'s utility function
%over strategy profiles
%is the function
%$U_i^\stratset :  \stratset \longrightarrow \{ 0,1 \}  $
%such that,
%for all $\strat \in \stratset$:
%\begin{align*}
%U_i^\stratset ( \strat ) = 1 & \text{ iff }  U_i^\histset ( H_{S_0, \strat} ) = 1
%\end{align*}
%
%
%
%
%\end{definition}

%%%%%%%%%%%%%%%%%%%%
\subsection{Solution concepts}\label{sec:solutions}
%%%%%%%%%%%%%%%%%%%%

%In this paper we focus on two simple solution concepts that
%pertain to the agent analyssis of the game
%
We start with the concept of winning uniform strategy.
Intuitively speaking,
$\strat_i$ is a winning uniform
strategy for player $i$
if and only if $i$ knows that, by playing this strategy, she
will achieve her goal no matter what the other players
will decide to do.

\begin{definition}[Winning strategy]
Let
$ \infgame$
be an influence game
and let $\strat_i$ be
a strategy for player $i$.
We say that $\strat_i$
is a winning strategy for player $i$
if and only if
\begin{align}\label{formula:winning}
H_{S_0 , (\strat_i,\prof{\strat}_{-i})  } \models \gamma_i
\end{align}
for all profiles $\prof{\strat}_{-i}$ of strategies of players other than $i$.
A winning strategy is called \emph{uniform} if $(\ref{formula:winning})$ is true for all states $S \in S_0^{\sim_i}$.
\end{definition}

Observe that a winning strategy is not necessarily winning uniform, as the private state of an agent is not necessarily accessible to the other players.

\begin{example}
Let Ann, Bob and Jesse be three agents.
Let $p$ be an issue, and suppose that $B_{Ann}(p) = 1$,  $B_{Bob}(p) = 0$, $B_{Jesse}(p) = 0$.
Ann influences Bob and Jesse, while Bob influences Jesse as shown in the following picture:
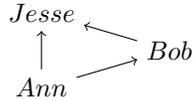
\begin{center}
\begin{tikzpicture}
\node (ann) at (0,0) {$Ann$};
\node (bob) at (1.7,0.5) {$Bob$};
\node (jesse) at (0,1) {$Jesse$};
%\node (al) at (1,0) {$l$};
\draw [->]  (ann) -- (bob) ;
\draw [->]  (ann) -- (jesse) ;
\draw [->]  (bob) -- (jesse) ;
%\draw [->]  (ai) -- (al) ;
\end{tikzpicture}
\end{center}
Suppose that the goal of Ann is $\eventually{\henceforth{ \bel{Jesse}{p}}}$.
Her winning (uniform) strategy is $\reveal{p}$ in all states: Bob will be influenced to believe $p$ in the second stage, and subsequently Jesse will also do so, since her influencers are unanimous (even if Bob plays $\hide{p}$).

%Ann influences Bob and Jesse, Bob influences Jesse.
%Now the goal of Ann is $FG \bel{Jesse}{1}$.
%She has the winning strategy of revealing $p$, so that Bob is influenced, and then in
%all cases Jesse will move to 1.
%Notice that if we delete the arrow from Ann to Jesse then this
%strategy is not winning anymore (it will be weakly dominant).
\end{example}

As we will show in the following section, the concept of winning strategy is too strong for our setting.
Let us then define the less demanding notion of weak dominance:

\begin{definition}\label{def:wds}
Let
$ \infgame$
be an influence game
and let $\strat_i$ be
a strategy for player $i$.
We say that $\strat_i$
is a \emph{weakly dominant strategy} for player $i$ and initial state $S_0$
if and only if for all profiles $\prof{\strat}_{-i}$ of strategies of players other than $i$
and for all strategies $\strat_i'$ we have:
\begin{align}\label{formula:wds}
H_{S_0 , (\strat_i',\prof{\strat}_{-i})  } \models \gamma_i \Rightarrow H_{S_0 , (\strat_i,\prof{\strat}_{-i})} \models \gamma_i
\end{align}
A weakly dominant strategy is called \emph{uniform} if $(\ref{formula:wds})$ is true for all initial states $S \in S_0^{~_i}$.
\end{definition}

\begin{example}
%Let there be four agents and one issue.
%$B_2=B_3=0$ and $B_1=B_4=1$, and $V_i=1$
%The $\gamma_2=FG \bel{i}{p}$
%
%1 is connected to 2, 2 to 4, 4 to 1, 2 to 3.
%
%Now if 2 plays not hide, then its goal will not be satisfied. So play hide is a winning strategy.
Let us go back to the previous example and suppose now that Ann still believes $p$, but does not influence Jesse any longer.
%\begin{center}
%\begin{tikzpicture}
%\node (ann) at (0,0) {$Ann$};
%\node (bob) at (1,0) {$Bob$};
%\node (jesse) at (0,1) {$Jesse$};
%%\node (al) at (1,0) {$l$};
%\draw [->]  (ann) -- (bob) ;
%%\draw [->]  (ann) -- (jesse) ;
%\draw [->]  (bob) -- (jesse) ;
%%\draw [->]  (ai) -- (al) ;
%\end{tikzpicture}
%\end{center}
In this case, Ann does not have a winning strategy: if neither Bob nor Jesse do not believe $p$, it is sufficient for Bob to play $\reveal{p}$ to make sure that she will never satisfy her goal.
However, the strategy $\reveal{p}$ is a weakly dominant strategy for Ann.
\end{example}

Now let us consider the following concept of best response.
Intuitively speaking, $\strat_i$ is a
best response to $\prof{\strat}_{-i}$
if and only if player $i$
knows that the worst she could possibly
get by playing
$\strat_i$, when
the others play
$\prof{\strat}_{-i}$,
is better or equal to the worst she could possibly get by
playing a strategy different from
$\strat_i$.
%To simplify notation, we introduce a fictitious utility $U_i(S)= 1$ for states $S$
%that satisfy the goal of agent $i$, and $U_i(S)= 0$ for states where the goal is not satisfied:

\begin{definition}[Best response]\label{def:bestresponse}
Let $\infgame$ be an influence game, let $\strat_i \in \stratset_{i}$
and  let $\prof{\strat}_{-i} \in \stratset_{-i}$.
We say that $\strat_i$
is a best response to  $ \prof{\strat}_{-i}$ wrt. initial state $S_0$
if and only if for all  $\strat_{i}' \in \stratset_{i}$:
\begin{multline*}
(\forall S\in  S_0^{\sim_i} \; (H_{S , (\strat_i,\prof{\strat}_{-i})  } \models \gamma_i) \text{ or } \\
   (\exists S,S'\in  S_0^{\sim_i}\;H_{S , (\strat_i,\prof{\strat}_{-i})  } \not \models \gamma_i  \text{ and }
                 H_{S' , (\strat_i',\prof{\strat}_{-i})  } \not\models \gamma_i  ))
\end{multline*}
%\begin{align*}
%\min_{    S\in  S_0^{\sim_i}      } U_i (H_{S , (\strat_i,\strat_{-i})  }) \geq
%\min_{    S\in  S_0^{\sim_i}       } U_i (H_{S , (\strat_i',\strat_{-i})  }).
%\end{align*}
%$$\exists S\in S_0^{\sim_i}  \text{ s.t. }H_{S , (\strat_i,\strat_{-i})} \models \gamma_0 \Rightarrow \exists S'\in S_0^{\sim_i} \text{ s.t. }  H_{S' , (\strat_i',\strat_{-i})}  \models \gamma_0$$
\end{definition}
If we rephrase this definition through some utility notion, then we can consider a fictitious utility $U_i(S)= 1$ for states $S$  that satisfy the goal of agent $i$, and $U_i(S)= 0$ for states where the goal is not satisfied. In that case, our definition of best response corresponds to $\min_{    S\in  S_0^{\sim_i}      } U_i (H_{S , (\strat_i,\prof{\strat}_{-i})  }) \geq \min_{    S\in  S_0^{\sim_i}       } U_i (H_{S , (\strat_i',\prof{\strat}_{-i})  })$.
This definition is justified on the basis of the prudential criterion according to which,
if an agent
does not have a probability distribution
over the set of possible states,
she should
focus on the worst possible
outcome and choose the action
 whose worst possible outcome is at least as good as the worst possible outcome of any other actions (see, e.g., \cite{Wald,BoutilierKR,BrafmanTennenholtz}).
%\el{Modified this paragraph. Added justification prudential principle.}

Definition~\ref{def:bestresponse} allow us to define the concept of Nash equilibrium
for games with incomplete information such as influence games:

\begin{definition}[Nash equilibrium]
Let
$ \infgame$
be an influence game
and let $\prof{\strat}$ be a strategy profile.
$\prof{\strat}$ is a Nash equilibrium
if and only if, for all $i \in \N$,
$\strat_i$
is a best response to  $\prof{\strat}_{-i}$.
\end{definition}

%defined in Section~\ref{sec:solutions}.

%%%%%%%%%%%%%%%%%%%%%%%%%%%%%%%%%%%%%%%%%%%%%%%%%%%%%%%%%%%%%%%%%%%%%%%%
\subsection{Influence network and solution concepts}
%%%%%%%%%%%%%%%%%%%%%%%%%%%%%%%%%%%%%%%%%%%%%%%%%%%%%%%%%%%%%%%%%%%%%%%%

In this section we show some preliminary results about the interplay between the network structure and the existence of solutions concepts.
In what follows we only consider influence games where the aggregation function is the unanimous one (see Definition~\ref{def:agg}). In the interest of space, most proofs will only be sketched.
Let us first give the following:

\begin{definition}
A goal $\gamma$ is coherent with an initial state $S_0$ in game $\infgame$ if and only if there exists a strategy profile $\strat$ inducing history $H$ such that $H_{S_0, \strat}\models \gamma$.
\end{definition}

Clearly, if $\gamma_i$ is not coherent with initial state $S_0$, then all strategies for player $i$ are equivalent.
The following lemma shows that visibility goals cannot be enforced by means of a winning strategy:

\begin{proposition}
If $\gamma_i$ entails one formula of the form $\visatom{j}{p}$, $\know{j}{\bel{i}{p}}$, or $\nextop{\visatom{j}{p}}$ for $j\not = i$, with belief and visibility atoms eventually negated, then $i$ does not have a winning strategy.
\end{proposition}

To see this, consider that if an individual goal $\gamma_i$ concerns the visibility of another agent about a given issue $p$, then this second agent can always respond $\hide{p}$ and make sure that $\gamma_i$ is false.

%Let $\language_J$ be the $\language_J$ only about $J$. Clearly, since individuals outside the subtree generated by the
%
%\begin{lemma}
%
%\end{lemma}

Let us now introduce a simpler language for goals, in order to study the limitations of considering winning strategies in this setting.
Let $\logicgoal_{\mathcal J}$ be the language of future goals about a subset of agent $\mathcal J \subseteq \N$, which focuses on the future opinions of agents in $\mathcal J$ without considering the visibility. This language is defined by the following BNF:

\begin{center}\begin{tabular}{lcl}
 $\alpha$  & $\bnf$ & $\bel{j}{p} \mid \neg\bel{j}{p}  \mid \alpha \wedge \alpha$\\
 $\varphi$  & $\bnf$ & $\nextop{\alpha}  \mid \nextop{\varphi} \mid \henceforth{\varphi} \mid \eventually{\varphi}$\
\end{tabular}\end{center}

\noindent
where $j\in\mathcal J$ and $p$ ranges over $\I$.

Let us introduce some further notation. If $i,j\in \N$ we say that $i$ \emph{controls} $j$ if either $\Inf(j)=\{j\}$, or for all paths $\l_1,\dots,l_n$ such that $(l_i, l_{i+1})\in E$, $l_1=i$ and $l_n=j$, we have that $i$ controls each $l_k$. We can now prove the following:

\begin{proposition}\label{prop:winning}
If $\gamma_i\in\logicgoal_{\mathcal J}$ then $i$ has a winning strategy for all initial states $S_0$ if and only if $i$ controls $j$ for all $j\in \mathcal J$.
\end{proposition}

\begin{proof}
One direction is easier: if $i$ controls agent $j$, then her winning strategy is to always play $\reveal{p}$ in case her goal is consistent with her opinion, e.g. if her goal is $\bel{j}{p}$ and $ H_{i,0}^B (p) =1 $. Otherwise always playing $\hide{p}$ guarantees that her goal will be satisfied. Note that the consistency of $\gamma_i$ is crucial here.

For the other direction, consider a network in which $j$ has more than one influencer, say $k$, which is however not controlled by $i$.
A simple case study shows that there always exists an initial state in which $i$ does not have a winning strategy.
For instance, if $\gamma_i=\nextop{\bel{j}{p}}$, then in an initial state in which $B_j(p)=0$ it is sufficient for agent $k$ to play $\reveal{p}$ to make sure that agent $j$ never updates her belief, and hence that $\gamma_i$ will not be satisfied.
\end{proof}

Proposition~\ref{prop:winning} shows that the concept of winning strategy is too strong in influence games, as it can only be applied in situations in which an agent has exclusive control over the opinion of another.

Let us now focus on a particular influence game, in which the agents' goal is to reach a consensus about $p$.
%, where with consensus we mean a configuration where all agents share the same opinion about $p$.
That is, each agent $i$ adopts the following goal:
%about other agents (excluding itself as no agent can influence itself):
%represented by the two $\logicgoal_\N$ formulas for handling positive and negative opinions:

\begin{alignat*}{2}\label{eqn:consensus}
 \gamma_i^{+} =_{\mathit{def}} \eventually\nextop{\left(\bigwedge_{j\ne i} \bel{j}{p}\right)} \qquad &&
\gamma_i^- =_{\mathit{def}} \eventually\nextop{\left(\bigwedge_{j\ne i} \neg\bel{j}{p}\right)}
\end{alignat*}
%\bigwedge_{j\ne i}\neg \bel{j}{p}%
Consensus means that the conjunction of each individual goal leads to a state where all agents have $p$ as opinion.
%Let us focus on the unanimity $F_i$ as defined in Definition~\ref{def:agg}.

\begin{theorem}
If all agent $i$ has the goal of consensus represented by $\gamma_i^{+}$ (respectively  $\gamma_i^{-} $), and the network $E$ is fully connected, then for any initial state $S_0$ there always exists a Nash equilibrium $\widehat{\prof{\strat}}$ such that
$
H_{S_0 , \widehat{\prof{\strat}}  } \models \bigwedge_{i\in\N} \gamma_i^+
$ (respectively, $\gamma_i^-$).
Moreover, each strategy is weakly dominant.
%\[
%H_{S_0 , \widehat{\strat}  } \models\eventually \nextop{\left(\bigwedge_j \bel{j}{p} \lor \bigwedge_j \neg\bel{j}{p}\right)}
%\].
\end{theorem}

\begin{proof}
Take an initial state $S_0 = (\B, \prof{\vis})$, and assume that all individuals have goal $\gamma_i^+$.
%Let us focus on the case where consensus has to be reached at the next state: $\nextop{(\bigwedge_j \bel{j}{p})}$
Consider the four possible states about $p$ for agent $i$: either $p$ is true or false, and $p$ is visible or not. Assume $B_i(p) = 1$, then, regardless of visibility, we show that strategy $\reveal{p}$ is weakly dominant.
%the best response to any strategy $\strat_{-i}$.
%To prove that, consider that  $\gamma_i$ focuses on the next state: $\nextop{(\bigwedge_{j\ne i} \bel{j}{p})}$.
For all $S\in  S_0^{\sim_i}$, either $H_{S, (\reveal{p},\strat_{-i}) } \models \bigwedge_{j\ne i} \bel{j}{p}$ or  $H_{S, (\reveal{p},\strat_{-i})  } \not\models \bigwedge_{j\ne i} \bel{j}{p}$;  for that former case, definition of unanimity entails that $H_{S, (\skipact, \strat_{-i}) } \not\models \bigwedge_{j\ne i} \bel{j}{p}$ and $H_{S, (\hide{p}, \strat_{-i}) } \not\models \bigwedge_{j\ne i} \bel{j}{p}$.
%Hence,  strategy $\reveal{p}$ is the best response for any $\strat{-i}$ as long as $B_i(p) = 1$.
In a similar way, we can show that strategy $\hide{p}$ is a best response to any strategy $\prof{\strat_{-i}}$ if $B_i(p) = 0$.
We then built up a strategy profile $\widehat{\prof{\strat}}$ w.r.t.\ $B_i$: $\strat_i = \reveal{p}$ if $B_i(p) = 1$ otherwise $\strat_i = \hide{p}$, this strategy profile is a Nash equilibrium in weakly dominant strategies. %Every profile where only one individual plays "reveal" are Nash equilibrium.
\end{proof}

Observe that if the network $E$ is not fully connected this result does not hold since, for instance, the opinion of an isolated agent cannot be changed.
%
%Let us focus on the unanimity $F_i$ as defined in Definition~\ref{def:unanimity}.
%
%As previously show, it is almost impossible to characterize winning strategies in an influence game. Hereafter, we
Let us now focus on specific shapes of the influence graph where weakly dominant strategy exists.
Those strategies concern the sources %of the network, i.e. those
the agents who have no influencers.
%We omit the proof in the interest of space.

%\begin{theorem}
%If $E$ is a chain, $i^*$ is the source, $S_0$ is coherent with $\gamma_i$ and the goals of the people after you imply your goal, then there is a winning dominant strategy for $i$.
%\end{theorem}

\begin{theorem}
If $E$ is acyclic, $I^*$ is the set of sources, and $S_0$ is coherent with $\bigwedge_{i^*\in I^*} \gamma_i$, and each source \emph{influences only one individual}, then for all $i^*$ there exists a weakly dominant strategy.
\end{theorem}

\begin{proof}[sketch]
If $\bigwedge_{i^*\in I^*}\gamma_{i^*}$ is coherent with $S_0$ then there exists some induced history $H$ by some strategy $\prof{\strat}^c$ such that $H_{S_0, \prof{\strat}^c } \models \bigwedge_{i^*\in I^*}\gamma_{i^*}$.
Consider a source $i^*$ and its goal $\gamma_{i^*}$. 
Subformulas of its goal refer to agents that it directly influences or not. In the latter case all the strategies of the source will be equivalent (hence weak-dominant). 
If they talk about the (only) individuals that is influenced by the source, then by the monotonicity of the aggregation procedure it is weakly dominant to play $\reveal{p}$ if the source goal is coherent with the source's belief. And $\hide{p}$ otherwise. (A case study is required to obtain the full proof).
\end{proof}

This result, once more, shows the difficulty of playing an influence game. It is actually possible to exhibit examples of acyclic influence graphs with sources influencing multiple agents where  no weakly dominant strategies exist for these sources.

%%%%%%%%%%%%%%%%%%%%%%%%%%%%%%%%%%%%
\subsection{Computational complexity}
%%%%%%%%%%%%%%%%%%%%%%%%%%%%%%%%%%%%

In this section we exemplify the use of $\logic$ and the complexity results presented in Section~\ref{sec:goals} for the computation of strategic aspects of influence games.
We do so by providing a PSPACE algorithm to decide whether a strategy profile is a Nash equilibrium.

Let $\textsc{MEMBERSHIP(F)}$ be the following problem: given as input a set of individuals $\N$, issues $\I$, goals $\gamma_i$ for $i\in \N$ -- which together with $F$ form an influence game $\infgame$ -- and a strategy profile~$\prof{Q}$, we want to know whether~$\prof{Q}$ is a Nash equilibrium of~$\infgame$.

The algorithm presented in \cite{GutierrezEtAlIC20015} in the setting of iterated boolean games cannot be directly applied to our setting for two reasons.
First, our histories are generated by means of an aggregation function $F$ that models the diffusion of opinions -- i.e., agents have a concurrent control on a set of propositional variables.
Second, not all conceivable strategies are available to players, as we focus on state-based strategies.

We therefore begin by translating a state-based strategy in $\logic$.
%We sketch most of the details in the interest of space.
Clearly, a conjunction of literals $\alpha(S)$ can be defined to uniquely identify a state $S$:
$\alpha(S)$ will specify the private opinion of all individuals and their visibility function.
Given action $a$, let
%$\beta_i(a)=\nextop{\visatom{i}{p}}$ if $a=\reveal{p}$, $\beta_i(a)=\nextop{\neg\visatom{i}{p}}$ if $a=\hide{p}$, and $\beta_i(a)=\top$ if $a=\skipact$.

$$\beta_i(a)=
 \begin{cases}
 \nextop{\visatom{i}{p}} & \text{ if }  a=\reveal{p} \\
 \nextop{\neg\visatom{i}{p}} & \text{ if }a=\hide{p} \\
 \top & \text{ if } a=\skipact
 \end{cases}
$$
\noindent
We can now associate a $\logic$ formula to each strategy $Q_i$:
% in the following way:
$$\tau_i(Q_i)=_{\mathit{def}}  \bigwedge_{S\in\stateset} \alpha(S) \rightarrow \beta_i(Q_i(S))$$

If $\prof{Q}$ is a strategy profile, let $\tau(\prof{Q})=\bigwedge_{i\in\N} \tau_i(Q_i)$.
We now need to encode the aggregation function into a formula as well. Recall the unanimous issue-by-issue aggregation function of Definition~\ref{def:agg} and consider the following formulas $\text{\it unan}(i,p)$:
\begin{align*}
& \nextop{\ \bel{i}{p}} \leftrightarrow \\
& \Big(
\big[\bigwedge_{j \in \Inf(i)}
\mathsf{X} \neg \visatom{j}{p} \wedge \bel{i}{p} \big] \vee \\
& \big[\bigvee_{j \in \Inf(i)}
\mathsf{X} \ \visatom{j}{p}
 \wedge
 \bigwedge_{j \in \Inf(i)}
 ( \mathsf{X} \ \visatom{j}{p} \rightarrow \bel{j}{p} ) \big] \vee \\
 & \big[\bigvee_{j,z \in \Inf(i):}
 ( \mathsf{X} \ \visatom{j}{p} \wedge \mathsf{X} \ \visatom{z}{p} \wedge \\
 &  \bel{j}{p} \wedge \neg \bel{j}{p}) \wedge \bel{i}{p} \big]
  \Big)
 \end{align*}
\noindent
as well as the following formula $\text{\it unan}(i,\neg p)$:
\begin{align*}
& \nextop{\ \neg \bel{i}{p}} \leftrightarrow\\
& \Big(
\big[\bigwedge_{j \in \Inf(i)}
\mathsf{X} \neg \visatom{j}{p} \wedge  \neg \bel{i}{p} \big]\vee \\
&
\big[\bigvee_{j \in \Inf(i)}
\mathsf{X} \ \visatom{j}{p}
 \wedge
 \bigwedge_{j \in \Inf(i)}
 ( \mathsf{X} \ \visatom{j}{p} \rightarrow  \neg \bel{j}{p} ) \big] \vee \\
 &
 \big[\bigvee_{j,z \in \Inf(i):}
 ( \mathsf{X} \ \visatom{j}{p} \wedge \mathsf{X} \ \visatom{z}{p} \wedge \\
&   \neg \bel{j}{p} \wedge  \bel{j}{p}) \wedge  \neg \bel{i}{p} \big]
  \Big)
\end{align*}
\noindent
This formula ensures that if the influencers of agent $i$ are unanimous, then agent $i$'s opinion should be
defined according to the three cases described in Definition~\ref{def:agg}. Recall that, while actions take one time unit to be effectuated (hence the $\nextop$ operator in front of $\visatom{j}{p}$), the diffusion of opinions is simultaneous.
Let now:
$$\tau(F^U_i)=_{\mathit{def}}\bigwedge_{\{i \in \N \mid \Inf(i)\not=\emptyset\}}  \bigwedge_{\{p \in \I\}} ( \text{\it unan}(i,p) \wedge  \text{\it unan}(i,\neg p) ) $$

\noindent
This formula encodes the transition process defined by the opinion diffusion.
$\tau(F^U_i)$ is polynomial in both the number of individuals and the number of issues (in the worst case it is quadratic in  $n$ and linear in $m$).
We are now ready to prove the following result:

\begin{theorem}\label{thm:nash}
$\textsc{MEMBERSHIP}(F_i^U)$ is in PSPACE.
\end{theorem}

\begin{proof}
Let $\prof{Q}$ be a strategy profile for game $\infgame$. The following algorithm can be used to check whether $\prof{Q}$ is a Nash equilibrium. For all individuals $i\in \N$, we first check the following entailment:
\begin{align*}
\tau(\prof{Q}) \wedge \tau(F_i) \models_{ \mathsf{LTL}  } \redu (\gamma_i)
\end{align*}
in the language of $\mathsf{LTL}$
built out the set of atomic
propositions   $\{    \bel{i}{p} :  i \in \N  \text{ and } p \in \I      \}
\cup \{    \visatom{i}{p} :  i \in \N  \text{ and } p \in \I      \}$,
where $\redu(\gamma_i)$ is defined as in Algorithm
\ref{algo:redproc} in Section \ref{sec:reduction}.
%\el{Modified here. I have replaced $\gamma_i$
%by $\redu (\gamma_i)$ since $\gamma_i$
%is not part of the $\mathsf{LTL}$ language. Notice that, due to changes
%I have made in Section 3.2, $\redu(\varphi)$ is now
%an LTL formula.
%}

If this is not the case, we consider all the possible strategies $Q'_i\not = Q_i$ for agent $i$ -- there are exponentially many, but each one can be specified in space polynomial in the size of the input -- and check the following entailment:
\begin{align*}
\tau(\prof{Q}_{-i}, Q_i) \wedge \tau(F_i) \models_{ \mathsf{LTL}  } \redu(\gamma_i)
\end{align*}

If the answer is positive we output NO, otherwise we proceed until all strategies and all individuals have been considered. The entailment for
$\mathsf{LTL}$
can be reduced to the problem
of checking validity in $\mathsf{LTL}$. Indeed, the following equivalence holds:
\begin{align*}
\psi \models_{ \mathsf{LTL}  } \varphi \text{ iff }   \models_{ \mathsf{LTL}  } \henceforth \psi \rightarrow \varphi
\end{align*}
Since the
problem
of checking validity in $\mathsf{LTL}$
 can be solved in PSPACE \cite{SistlaClarkeLTL}, we obtain the desired upper bound.
\end{proof}

We conjecture that the problem is also PSPACE-complete, as a reduction in line with the one by \cite{GutierrezEtAlIC20015} is likely to be obtained.

Observe that Theorem~\ref{thm:nash} can easily be generalised to all aggregation procedures that can be axiomatised by means of polynomially many $\logic$ formulas -- with the eventual use of Lemma~\ref{redprocedure} to translate $\logic$-formulas in $\mathsf{LTL}$.
This is not the case for all aggregation procedures: the majority rule -- i.e., the rule that updates the opinion of an individual to copy that of the majority of its influencers -- would for instance require an exponential number of formulas, one for each subset of influencers that forms a relative majority.
The study of the axiomatisation of aggregation procedures for opinion diffusion constitutes a promising direction
for future work.

%\el{Well, what I do like in this section is that we do not use the result about complexity of model checking in Section 3.2. }

%%%%%%%%%%%%%%%%%%%%%%%%%%%%%%%%%%%
\section{Conclusions and future work}\label{sec:conclusions}
%%%%%%%%%%%%%%%%%%%%%%%%%%%%%%%%%%%%

In this paper we proposed a model, inspired from related work on iterated boolean games \cite{GutierrezEtAlIC20015},
that allows us to explore some basic aspects of strategic reasoning in social influence.
We grounded our model on related work \cite{SchwindEtAlAAAI2015,GrandiEtAlAAMAS2015},
 which modelled the process of social influence by means of aggregation procedures from either
 judgment aggregation or belief merging, and
we augmented it with the introduction of a simple logical language for the expression of
temporal and epistemic goals.
This allowed us to inquire into the multiple aspects of the relation between the structure
of the influence network, and the existence of well-known game-theoretic solution concepts.
Moreover, we were able to show that model checking for our language, as well
as the problem of checking whether a given profile is a Nash equilibrium, is in PSPACE,
hence no harder than the linear temporal logic on which our language is based.

There are multiple directions in which this work can be expanded.
First, the introduction of extra actions to add or sever trust links
may add an important dynamic aspect to the network structure.
Second, we may allow agents to lie about their private preferences, hence providing them
with more strategies to attain their goals.
Third, to develop our framework to its full generality we could introduce integrity constraints
among the issues at hand.
%but is likely to create troubles in the interpretation of the setting and in many of the strategic results.
In all these cases, a deeper study of the interconnection between the network structure and the
strategies played by the agents of extreme interest, and has the potential to
unveil general insights about the problem of social influence.
%as well as the axiomatisability
%of diffusion rules other than the unanimity rule presented in this paper.

\bibliographystyle{abbrv}
\bibliography{pop}

\end{document}